\documentclass[10pt]{article}
\usepackage{amsmath}
\usepackage{amssymb, amscd, amsthm}
\usepackage[all]{xy}
\usepackage[dvips]{graphicx}
\usepackage{verbatim}
\usepackage[perpage,symbol*]{footmisc}
\newtheorem{theorem}{Theorem}

\newtheorem{lemma}{\textbf{Lemma}}[section]
\newtheorem{remark}{\textbf{Remark}}[section]
\newtheorem{corollary}{\textbf{Corollary}}[section]

\usepackage[justification=centering]{caption}
\usepackage{longtable}
\newcommand{\tabincell}[2]{\begin{tabular}{@{}#1@{}}#2\end{tabular}}
\usepackage{multirow}
\usepackage{slashbox}

\newcommand{\F}{\mathbb{F}}

\begin{document}

\baselineskip 17pt
\title{\Large\bf Constructions of MDS Self-dual Codes from Short Length}

\author{\large Derong Xie \quad\quad Xiaolei Fang \quad\quad Jinquan Luo*}\footnotetext{The authors are with School of Mathematics
and Statistics \& Hubei Key Laboratory of Mathematical Sciences, Central China Normal University, Wuhan China 430079.\\
 E-mail: derongxie@yahoo.com(D.Xie), fangxiaolei@mails.ccnu.edu.cn(X.Fang), luojinquan@mail.ccnu.edu.cn(J.Luo)}

\date{}
\maketitle

{\bf Abstract}: Systematic constructions of MDS self-dual codes is widely concerned. In this paper, we consider the constructions of MDS Euclidean self-dual codes from short length. Indeed, the exact constructions of MDS Euclidean self-dual codes from short length ($n=3,4,5,6$) are given. In general, we construct more new of $q$-ary MDS Euclidean self-dual codes from MDS self-dual codes of known length via generalized Reed-Solomon (GRS for short) codes and extended GRS codes.


\section{Introduction}

Let $\mathbb{F}_{q}$ be the finite field of cardinality $q$ where $q$ is a power of some odd prime $p$. An $[n,k,d]_q$ linear code $\mathcal{C}$ over $\F_q$ is a
linear $\F_q$-subspace of $\F_q^n$ with dimension $k$ and minimal (Hamming) distance $d$.  The Singleton bound states that $n\geq k+d-1$. The code $\mathcal{C}$ is called maximum distance separable(MDS) if the parameters can reach the Singleton bound.  The set
$$\mathcal{C}^{\bot}=\left\{x\in\mathbb{F}_{q}^{n} \ \Big| \ (x,c)=0 \ \mbox{for \  all} \ c\in\mathcal{C}\right\}$$
is called the dual code of $\mathcal{C}$ where $(x,c)$ is the usual  inner product in $\mathbb{F}_{q}^{n}$. If $\mathcal{C}\subseteq \mathcal{C}^{\bot}$, then $\mathcal{C}$ is called  self-orthogonal. If $\mathcal{C}=\mathcal{C}^{\bot}$, then $\mathcal{C}$ is called  self-dual.

MDS  self-dual codes attract much attention since it has good algebraic structure and optimal parameters. MDS  self-dual codes of length $n$ over $\mathbb{F}_{2^{m}}$ have been completely determined  in \cite{GG}. Furthermore, over finite field of odd prime characteristic, MDS  self-dual codes  are constructed via different techniques, (1). orthogonal designs, see (\cite{GK}, \cite{HK1}, \cite{HK2}); (2). building up technique, see (\cite{KL1}, \cite{KL2}); (3). constacyclic codes, see (\cite{KZT}, \cite{TW}, \cite{YC}); (4). (extended) GRS codes, see (\cite{FF3}, \cite{GKL}, \cite{JX2}, \cite{LLL}, \cite{TW}, \cite{Yan}, \cite{ZF}). We list all the known results on the systematic constructions of MDS  self-dual codes, which are depicted in Table 1.

\begin{center}
\begin{longtable}{|c|c|c|}  %table 1
\caption{Some known systematic construction on MDS  self-dual codes of length
$n$\\ ($\eta$ is the quadratic character of $\mathbb{F}_{q}$) } \\ \hline
$q$ & $n$ even & Reference\\  \hline
$q$ even  &  $n \leq q$   & \cite{GG} \\ \hline
$q$ odd & $n=q+1$ & \cite{GG}\\ \hline
$q$ odd & $(n-1)|(q -1)$, $\eta(1 - n) = 1$ &   \cite{Yan} \\ \hline
$q$ odd & $(n-2)|(q - 1)$, $\eta(2 - n) = 1$ &   \cite{Yan}\\ \hline
%% $q = r^{t}$, $t$ even  &  $(n - 1)|(r - 1)$  &  \cite{GUE} \hline
$q = r^{s}$ , $r$ odd, $s\geq 2$ & $n = lr$,  $l$ even and $2l|(r - 1)$ &   \cite{Yan} \\ \hline

$q = r^{s}$ , $r$ odd, $s \geq 2$ & $n = lr$,  $l$ even , $(l - 1)|(r - 1)$ and $\eta(1 - l)=1$ &   \cite{Yan} \\ \hline

$q = r^{s}$ , $r$ odd, $s \geq 2$ & $n = lr + 1$, $l$ odd , $l|(r - 1)$ and $\eta(l) = 1$  &   \cite{Yan} \\ \hline
 $q = r^{s}$ , $r$ odd, $s \geq 2$ & $n = lr + 1$, $l$ odd , $(l - 1)|(r - 1)$ and $\eta(l - 1) = \eta(-1) = 1$ &  \cite{Yan} \\ \hline

$q=r^2$  & $n \leq r$  & \cite{JX2} \\ \hline
$q = r^2, r\equiv3\,(\mathrm{mod}\,4)$  &  $n= 2tr$ for any $t \leq \frac{r - 1}{2}$ &   \cite{JX2}\\ \hline

$q = r^2$, $r$ odd & $n = tr$, $t$ even and $1 \leq t \leq r$ &   \cite{Yan} \\ \hline

 $q = r^2$, $r$ odd & $n = tr + 1$,  $t$ odd and $1 \leq t \leq r$ &   \cite{Yan} \\ \hline

$q \equiv1\,(\mathrm{mod}\,4)$ &  $ n|(q - 1), n < q - 1$ &   \cite{Yan}\\ \hline
$q\equiv1\,(\mathrm{mod}\,4)$ &  $4^{n}\cdot n^{2} \leq q$ &  \cite{JX2}\\ \hline

  $q = p^k $, odd prime $p$ & $n= p^r+1$, $r|k$ &   \cite{Yan} \\ \hline
$q = p^k $, odd prime $p$ & $n= 2p^e$, $1 \leq e < k$, $\eta(-1) = 1$&  \cite{Yan} \\ \hline
$q=r^2$, $r$ odd & $n=tm$, $1\leq t \leq \frac{r-1}{\gcd(r-1,m)}$, $\frac{q-1}{m}$ even &  \cite{LLL} \\ \hline
$q=r^2$, $r$ odd & $n=tm+1$, $tm$ odd, $1\leq t \leq \frac{r-1}{\gcd(r-1,m)}$ and $m|(q-1)$  & \cite{LLL}\\ \hline
$q=r^2$, $r$ odd & $n=tm+2$, $tm$ even, $1\leq t \leq \frac{r-1}{\gcd(r-1,m)}$ and $m|(q-1)$   &   \cite{LLL}\\\hline
$q=r^2$, $r$ odd & $n=tm$, $1\leq t \leq \frac{r+1}{\gcd(r+1,m)}$, $\frac{q-1}{m}$ even & \cite{FLLL} \\ \hline

$q=r^2$, $r$ odd  &\tabincell{c}{$n=tm+2$, $tm$ even(except when $t$ is even, $m$ is even\\
 and $r\equiv1\,(\mathrm{mod}\,4)$), $1\leq t \leq \frac{r+1}{\gcd(r+1,m)}$ and $m|(q-1)$}   &  \cite{FLLL}  \\\hline
 $q=r^2$, $r$ odd & $n=tm+1$, $tm$ odd, $2\leq t \leq \frac{r+1}{2\gcd(r+1,m)}$ and $m|(q-1)$  & \cite{FLLL} \\ \hline
   $q=r^2$, $r$ odd & \tabincell{c}{$n=tm$, $1\leq t \leq \frac{s(r-1)}{\gcd(s(r-1),m)}$, $s$ even, $s|m$,\\ $\frac{r+1}{s}$ even and $\frac{q-1}{m}$ even}  & \cite{FLLL} \\ \hline
      $q=r^2$, $r$ odd & \tabincell{c}{$n=tm+2$, $1\leq t \leq \frac{s(r-1)}{\gcd(s(r-1),m)}$, $s$ even, $s|m$,\\ $s\mid r+1$ and $m|(q-1)$}  & \cite{FLLL} \\ \hline
$q=p^{m}$, $m$ even, odd prime $p$ & $n=2tr^l$ with $r=p^s$, $s\mid\frac{m}{2}$, $0\leq l\leq \frac{m}{s}$ and $1\leq t\leq\frac{r-1}{2}$ & \cite{FF3}\\ \hline
$q=p^{m}$, $m$ even, odd prime $p$ &\tabincell{c}{$n=(2t+1)r^l+1$ with $r=p^s$, $s\mid\frac{m}{2}$, $0\leq l<\frac{m}{s}$ \\and $0\leq t\leq\frac{r-1}{2}$ or $l=\frac{m}{s}$, $t=0$} & \cite{FF3}\\ \hline
$q=p^m\equiv1\,(\mathrm{mod}\,4)$ & $n= p^l+1$ with $0\leq l\leq m$ &   \cite{FF3} \\ \hline
$q=r^2$, $r\equiv1\,(\mathrm{mod}\,4)$ and $s$ even & $n=s(r-1)+t(r+1)$ with $1\leq s\leq \frac{r+1}{2}$ and $1\leq t\leq \frac{r-1}{2}$ &   \cite{FLL} \\ \hline
$q=r^2$, $r\equiv3\,(\mathrm{mod}\,4)$ and $s$ odd & $n=s(r-1)+t(r+1)$ with $1\leq s\leq \frac{r+1}{2}$ and $1\leq t\leq \frac{r-1}{2}$ &   \cite{FLL} \\ \hline
\end{longtable}
 \end{center}

In this paper, we focus on constructions of MDS self-dual codes from short length. In Section 2, we introduce some basic definitions, notations and useful results. In Section 3, we show that new classes $q$-ary MDS self-dual codes is constructed via GRS codes and extended GRS coeds. Based on the union of affine subspaces, we give constructions of MDS self-dual codes of length $2tp^{l}$ (resp. $(2t+1)p^{l})+1$) from MDS self-dual GRS codes of length $2t$ (resp. MDS self-dual EGRS codes of length $2t+2$). Based on the union of cosets of some multiplicative subgroup, we give constructions of MDS self-dual codes of length $2te_{1}$ (resp. $(2t+1)e_{1}+1$) from MDS  self-dual codes of length $2t$ (resp. MDS self-dual EGRS codes of length $2t+2$) under certain conditions. Precisely, our main contribution is to construct new MDS self-dual codes of length $n$ as follows.
\begin{itemize}
  \item[] Let $q=p^{m}$ be an odd prime power, $s\mid m$ and $q-1=e_{1}e_{2}$.
  \item[(1)] $n=4p^{l}$ with $0\leq l < m$, if $p\equiv1 \ (\mbox{mod} \ 12)$ (see Theorem 1 (1)).
  \item[(2)] $n=6p^{l}$ with $0\leq l < m$, if $p\equiv1,9 \ (\mbox{mod} \ 40)$ (see Theorem 1 (2)).
  \item[(3)] $n=3p^{l}+1$ with $0\leq l < m$ if $p\equiv1,3 \ (\mbox{mod} \ 8)$ (see Theorem 2).
  \item[(4)] $n=5p^{l}+1$ with $0\leq l < m$ if $p\equiv1 \ (\mbox{mod} \ 24)$ (see Corollary 3.1).
  \item[(5)] $n=4e_{1}$, if $p\equiv1,3 \ (\mbox{mod} \ 8)$, $e_{1}$ odd and $e_{2}\geq4$ (see Theorem 4).
  \item[(6)] $n=(t+1)p^{sl}$ with $0\leq l < \frac{m}{s}$, if $t$ is odd, $t\mid (p^{s} -1)$ and $\eta(-t)= 1$ (see Theorem 5).
  \item[(7)] $n=(t+1)p^{sl}+1$ with $0\leq l < \frac{m}{s}$, if $t$ is even, $t\mid (p^{s} -1)$ and $\eta(t)=\eta(-1)= 1$ (see Theorem 5).
\end{itemize}
 Finally, we give a short summary of this paper in Section 4.

\section{Preliminaries}

For $1\leq n\leq q$, we choose $S=\{a_{1},a_{2},\ldots,a_{n}\}$ and $V=(v_{1},v_{2},\ldots,v_{n})$, where $a_{i}\in\mathbb{F}_{q}$ are distinct elements  and
$v_{i}\in\mathbb{F}_{q}^{*}$ ($v_{i}$ may not be distinct)  for $1\leq i\leq n$.
Then the generalized Reed-Solomom (GRS for short) code of length $n$ associated with $S$ and $V$ is
\begin{equation}\label{def GRS}
\mathbf{GRS}_{k}(S,V,q)=\left\{(v_{1}f(a_{1}),\ldots,v_{n}f(a_{n})):f(x)\in\mathbb{F}_{q}[x],\mathrm{deg}(f(x))\leq k-1\right\},
\end{equation}
for $1\leq k\leq n$.

Moreover, the extended generalized Reed-Solomom (EGRS for short) code associated with $S$ and $V$ is defined by:
\begin{equation}\label{def extended GRS}
\mathbf{EGRS}_{k}(S,V,q)=\left\{(v_{1}f(a_{1}),\ldots,v_{n}f(a_{n}),f_{k-1}):f(x)\in\mathbb{F}_{q}[x],
\mathrm{deg}(f(x))\leq k-1\right\},
\end{equation}
where $1\leq k\leq n$ and $f_{k-1}$ is the coefficient of $x^{k-1}$ in $f(x)$.

It is well known that $\mathbf{GRS}_{k}(S,V,q)$ and $\mathbf{EGRS}_{k}(S,V,q)$ are MDS code and their duals are also
MDS \cite{MS}.

Denote by 
\begin{equation*}
\Delta_{S}(a_{i})=\prod_{1\leq j\leq n,j\neq i}(a_{i}-a_{j}), \quad
f_{S}(x)=\prod_{1\leq j\leq n}(x-a_{j})
\end{equation*}
and $\eta$ the quadratic (multiplicative) character of $\mathbb{F}_{q}$ throughout this paper.
We give the following lemmas, which are useful in the proof of the main results.
\begin{lemma}\label{22}(\cite{ZF})
(1) Let $S=\{a_{1},a_{2},\cdots,a_{n}\}$ be a subset of $\mathbb{F}_{q}$ and $f_{S}(x)=\prod_{a\in S}(x-a)$. Then for any $a\in S$, $\Delta_{S}(a)=f^{'}_{S}(a)$.

(2) Let $S_{1}$ and $S_{2}$ be disjoint subsets of $\mathbb{F}_{q}$,  $S=S_{1}\cup S_{2}$. Then for $b\in S$,
\begin{equation*}
\Delta_{S}(b)=
\begin{cases}
\Delta_{S_{1}}(b)f_{S_{2}}(b), \text{ if $b\in S_{1}$,}\\
\Delta_{S_{2}}(b)f_{S_{1}}(b), \text{ if $b\in S_{2}$.}
\end{cases}
\end{equation*}

\end{lemma}

\begin{lemma}\label{fS}(\cite{ZF}) Let $a_{1},a_{2},\cdots,a_{n}$ be distinct elements in $\mathbb{F}_{q}$, $S=\{a_{1},a_{2},\cdots,a_{n}\}$. 
\begin{itemize}
  \item[(1)] (\cite{JX2}) Suppose that $n$ is even. There exists $V=(v_{1},v_{2},\cdots,v_{n})\in(\mathbb{F}_{q}^{*})^{n}$  such that  code $\mathbf{GRS}_{\frac{n}{2}}(S,V,q)$ is self-dual if and only if all $\eta(\Delta_{S}(a))$ are the same.
  \item[(2)] (\cite{Yan}) Suppose that $n$ is odd. There exists $V=(v_{1},v_{2},\cdots,v_{n})\in(\mathbb{F}_{q}^{*})^{n}$  such that code $\mathbf{EGRS}_{\frac{n+1}{2}}(S,V,q)$ is self-dual with length $n+1$ if and only if $\eta(-\Delta_{S}(a))=1$ for all $a\in S$.
\end{itemize}
\end{lemma}

\begin{lemma}\label{LeF}(\cite{FF3}) Suppose $q=p^{m}$ and $r=p^{s}$ with $s\mid m$. Fix an $\mathbb{F}_{r}$-subspace $H$ of $\mathbb{F}_{q}$ and an element
$\alpha\in\mathbb{F}_{q}\backslash H$. Label the elements of $\mathbb{F}_{r}$ as $\xi_{1},\xi_{2},\cdots,\xi_{r}$ and $0\leq z <r$. Denote $H_{i}=H+\xi_{i}\alpha$ for $0\leq i \leq z$. Then
\begin{itemize}
  \item[(i)] for any $\tau\in\mathbb{F}_{r}$, we have  $$f_{H}(\tau\alpha)=\tau f_{H}(\alpha);$$
  \item[(ii)] for any $0\leq i \leq z$ and $b\in H_{i},$ we have $$\Delta_{H_{i}}(b)=\Delta_{H}(0);$$
  \item[(iii)]  for any $0\leq i\neq j \leq z$ and $b\in H_{i},$ we have $$f_{H_{j}}(b)=(\xi_{i}-\xi_{j})f_{H}(\alpha).$$
\end{itemize}
\end{lemma}

\begin{lemma}\label{YZF}(\cite{ZF}) Let $\theta$ be a generator (primitive element) of $\mathbb{F}^{*}_{q}$ and $q-1=e_{1}e_{2}$. Denote $H=\langle\theta^{e_{2}}\rangle$. Then
$$f_{\theta^{i}H}(x)=x^{e_{1}}-\theta^{ie_{1}} \quad \mbox{and} \quad \Delta_{\theta^{i}H}(x)=f^{'}_{\theta^{i}H}(x)=e_{1}x^{e_{1}-1}.$$
\end{lemma}

\section{Main results}
In general, the construction of MDS  self-dual codes via GRS codes and EGRS codes comes down to the choice of $S$ satisfying Lemma \ref{fS}. In this section, we will take union of cosets satisfying Lemma \ref{fS}.

\begin{lemma}\label{Th1} Let $q=p^{m}$, $s\mid m$ and $H$ be an $\mathbb{F}_{p^{s}}$-subspace of $\mathbb{F}_{q}$ of dimension $l$. There exists a $q$-ary MDS  self-dual code of length $2tp^{sl}$ with $0\leq l < \frac{m}{s}$ provided that $\mathbf{GRS}_{t}(A,V,q)$ is self-dual for some $A=\{a_{1},a_{2},\cdots,a_{2t}\}\subseteq \mathbb{F}_{p^{s}}$ and $V\in(\mathbb{F}_{q}^{*})^{2t}$.
\end{lemma}
\begin{proof}
Let $\alpha\in\mathbb{F}_{q}\backslash H$ and $H_{i}=a_{i}\alpha+H$. Then label the elements of $B:=\bigcup\limits_{a_{i}\in A}(a_{i}\alpha+H)$ by $b_{1},b_{2},\cdots,b_{n}$. For $b_{i}\in H_{k}$ with $1\leq k\leq2t$, by Lemma \ref{22} and \ref{LeF},
\begin{equation}\nonumber
\begin{aligned}
\Delta_{B}(b_{i})&=\Delta_{H_{k}}(b_{i})\prod_{j=1,j\neq k}^{2t}f_{H_{j}}(b_{i})&\\
&=\Delta_{H}(0)(f_{H}(\alpha))^{2t-1}\prod_{j=1,j\neq k}^{2t}(a_{k}-a_{j})&\\
&=\Delta_{H}(0)(f_{H}(\alpha))^{2t-1}\Delta_{A}(a_{k}).&
\end{aligned}
\end{equation}
Note that all of $\eta(\Delta_{A}(a_{i}))$ are equal. By Lemma \ref{fS}, there exists a $q$-ary MDS  self-dual code of length $2tp^{sl}$ with $0\leq l < \frac{m}{s}$.
\end{proof}

\begin{remark} Let $q=p^{m}$ be an odd prime power.
\begin{itemize}
  \item[(1)] For $s\mid \frac{m}{2}$, we choose any $A\subseteq\mathbb{F}_{p^{s}}$ of size $2t$. Then there exists a $q$-ary MDS  self-dual code of length $2tp^{sl}$ with $0\leq l < \frac{m}{s}$ (see \cite{FF3}, Theorem 3.3(i)).
  \item[(2)] For $q\equiv1 \ (\mbox{mod} \ 4)$, we choose $A=\{0,1\}$ and $s=1$. Then there exists a $q$-ary MDS  self-dual code of length $2p^{l}$ with $0\leq l < m$ (see \cite{FF3}, Theorem 3.3(ii)).
\end{itemize}
\end{remark}

\begin{theorem}\label{GRS} Let $q=p^{m}$ be an odd prime power.
\begin{itemize}
  \item[(1)]  Suppose $p\equiv1 \ (\mbox{mod} \ 12)$. Then there exists a $q$-ary MDS  self-dual code of length $4p^{l}$ with $0\leq l < m$.
  \item[(2)] Suppose $p\equiv1,9 \ (\mbox{mod} \ 40)$. Then there exists a $q$-ary MDS  self-dual code  of length $6p^{l}$ with $0\leq l < m$.
\end{itemize}
\end{theorem}
\begin{proof}
\begin{itemize}
\item[(1)] For $p\equiv1 \ (\mbox{mod} \ 12)$, we choose $A=\left\{0,\frac{p-1}{3},\frac{2(p-1)}{3},p-1\right\}\subseteq \mathbb{F}_{p}$. Then
\begin{equation}\nonumber
\begin{aligned}
&\eta\left(\Delta_{A}(0)\right)=\eta\left(-\frac{p-1}{3}\cdot\frac{2(p-1)}{3}\cdot(p-1)\right)=\eta(2),&\\
&\eta\left(\Delta_{A}\left(\frac{p-1}{3}\right)\right)=\eta\left(\frac{p-1}{3}\cdot\frac{p-1}{3}\cdot\frac{2(p-1)}{3}\right)=\eta(-6),&\\
&\eta\left(\Delta_{A}\left(\frac{2(p-1)}{3}\right)\right)=\eta\left(-\frac{2(p-1)}{3}\cdot\frac{p-1}{3}\cdot\frac{p-1}{3}\right)=\eta(6),&\\
&\eta(\Delta_{A}(p-1))=\eta\left((p-1)\frac{2(p-1)}{3}\frac{p-1}{3}\right)=\eta(-2).&
\end{aligned}
\end{equation}
Note that $p\equiv 1 \ (\mbox{mod} \ 12)$. Then
\begin{itemize}
  \item[(i)] if $q$ is a square, then $\eta(x)=1$ for any $x\in\mathbb{F}^{*}_{p}$.
  \item[(ii)] if $q$ is a non-square, then $\eta(-1)=1$ and $\eta(3)=1$ by the quadratic reciprocity law.
\end{itemize}
Thus, by Lemma \ref{fS}, there exists $V=(v_{1},v_{2},v_{3},v_{4})\in(\mathbb{F}_{q}^{*})^{4}$  such that the code $\mathbf{GRS}_{2}(A,V,q)$ is self-dual.
By Lemma \ref{Th1}, there exists a $q$-ary MDS  self-dual code of length $4p^{l}$ with $0\leq l < m$.
 \item[(2)] For $p\equiv1,9 \ (\mbox{mod} \ 40)$, we choose
 $$A=\left\{0,\frac{p-1}{5},\frac{2(p-1)}{5},\frac{3(p-1)}{5},\frac{4(p-1)}{5},p-1\right\}\subseteq \mathbb{F}_{p}.$$ Then
\begin{equation}\nonumber
\begin{aligned}
&\eta\left(\Delta_{A}(0)\right)=\eta\left(-\frac{p-1}{5}\cdot\frac{2(p-1)}{5}\cdot\frac{3(p-1)}{5}\cdot\frac{4(p-1)}{5}
\cdot(p-1)\right)=\eta(6),&\\
&\eta\left(\Delta_{A}\left(\frac{p-1}{5}\right)\right)=\eta\left(\frac{p-1}{5}\cdot\frac{p-1}{5}\cdot\frac{2(p-1)}{5}
\cdot\frac{3(p-1)}{5}\cdot\frac{4(p-1)}{5}\right)=\eta(-30),&\\
&\eta\left(\Delta_{A}\left(\frac{2(p-1)}{5}\right)\right)=\eta\left(-\frac{2(p-1)}{5}\cdot\frac{p-1}{5}\cdot\frac{p-1}{5}
\cdot\frac{2(p-1)}{5}\cdot\frac{3(p-1)}{5}\right)=\eta(15),&\\
&\eta\left(\Delta_{A}\left(\frac{3(p-1)}{5}\right)\right)=\eta\left(\frac{3(p-1)}{5}\cdot\frac{2(p-1)}{5}\cdot\frac{p-1}{5}
\cdot\frac{p-1}{5}\cdot\frac{2(p-1)}{5}\right)=\eta(-15),&\\
&\eta\left(\Delta_{A}\left(\frac{4(p-1)}{5}\right)\right)=\eta\left(-\frac{4(p-1)}{5}\cdot\frac{3(p-1)}{5}\cdot\frac{2(p-1)}{5}
\cdot\frac{p-1}{5}\cdot\frac{p-1}{5}\right)=\eta(30),&\\
&\eta(\Delta_{A}(p-1))=\eta\left((p-1)\cdot\frac{4(p-1)}{5}\cdot\frac{3(p-1)}{5}\cdot\frac{2(p-1)}{5}\cdot\frac{p-1}{5}
\right)=\eta(-6).&
\end{aligned}
\end{equation}
Since $p\equiv 1,9 \ (\mbox{mod} \ 40)$,  $\eta(-1)=\eta(2)=\eta(5)=1$. Thus, by Lemma \ref{fS}, there exists $V=(v_{1},v_{2},\cdots,v_{6})\in(\mathbb{F}_{q}^{*})^{6}$  such that the code $\mathbf{GRS}_{3}(A,V,q)$ is self-dual.
By Lemma \ref{Th1}, there exists a $q$-ary MDS  self-dual code of length $6p^{l}$ with $0\leq l < m$.
 \end{itemize}
\end{proof}

\begin{lemma}\label{Th2} Let $q=p^{m}$, $s\mid m$ and $H$ be an $\mathbb{F}_{p^{s}}$-subspace of $\mathbb{F}_{q}$ of dimension $l$. Suppose $A=\{a_{1},a_{2},\cdots,a_{2t+1}\}\subseteq \mathbb{F}_{p^{s}}$ and $v\in(\mathbb{F}_{q}^{*})^{2t+1}$ such that $\mathbf{EGRS}_{t+1}(A,V,q)$ is self-dual. For $0\leq l < \frac{m}{s}$, if $q\equiv1 \ (\mbox{mod} \ 4)$ or $l$ even, then there exists a $q$-ary MDS  self-dual code of length $(2t+1)p^{sl}+1$.
\end{lemma}
\begin{proof}
Let $\alpha\in\mathbb{F}_{q}\backslash H$ and $H_{i}=a_{i}\alpha+H$. Then label the elements of $B:=\bigcup\limits_{a_{i}\in A}(a_{i}\alpha+H)$ by $b_{1},b_{2},\cdots,b_{n}$. For $b_{i}\in H_{k}$ with $1\leq k\leq2t+1$, by Lemmas \ref{22} and \ref{LeF}, then
\begin{equation}\nonumber
\begin{aligned}
\Delta_{B}(b_{i})&=\Delta_{H_{k}}(b_{i})\prod_{j=1,j\neq k}^{2t+1}f_{H_{j}}(b_{i})&\\
&=\Delta_{H}(0)(f_{H}(\alpha))^{2t}\prod_{j=1,j\neq k}^{2t+1}(a_{k}-a_{j})&\\
&=\Delta_{H}(0)(f_{H}(\alpha))^{2t}\Delta_{A}(a_{k}).&
\end{aligned}
\end{equation}
Divide the subspace $H$ into disjoint union $H=S\dot\cup(-S)\dot\cup\{0\}$. Then $\Delta_{H}(0)=(-1)^{\frac{p^{sl}-1}{2}}\left(\prod\limits_{x\in S}x\right)^{2}$ which implies that $\eta(\Delta_{H}(0))=\eta((-1)^{\frac{p^{sl}-1}{2}})=1$ holds for $q\equiv1 \ (\mbox{mod} \ 4)$ or $l$ even.
Thus, by Lemma \ref{fS},
$$\eta(-\Delta_{B}(b_{i}))=\eta(-\Delta_{A}(a_{k})(f_{H}(\alpha))^{2t}\Delta_{H}(0))=1$$
which implies that there exists a $q$-ary MDS  self-dual code of length $(2t+1)p^{sl}+1$ with $0\leq l < \frac{m}{s}$.
\end{proof}

\begin{remark} Let $q=p^{m}$ be an odd prime power.
\begin{itemize}
  \item[(1)] For $s\mid \frac{m}{2}$, we choose any $A\subseteq\mathbb{F}_{p^{s}}$ of size $2t+1$. Then there exists a $q$-ary MDS  self-dual code of length $(2t+1)p^{sl}+1$ with $0\leq l < \frac{m}{s}$ (see \cite{FF3}, Theorem 3.5(i)).
  \item[(2)] For $q\equiv1 \ (\mbox{mod} \ 4)$, we choose $A=\{1\}$ and $s=1$. Then there exists a $q$-ary MDS  self-dual code of length $p^{l}+1$ with $0\leq l < m$ (see \cite{FF3}, Theorem 3.5(ii)).
\end{itemize}
\end{remark}

\begin{theorem}\label{EGRS} Let prime $p\equiv1,3 \ (\mbox{mod} \ 8)$ and $q=p^{m}$. Then there exists a $q$-ary MDS  self-dual code of length $3p^{l}+1$ with $0\leq l < m$.
\end{theorem}
\begin{proof}
The condition $p\equiv1,3 \ (\mbox{mod} \ 8)$ implies $\eta(-2)=1$. We choose $A=\{0,\frac{p-1}{2},p-1\}\subseteq \mathbb{F}_{p}$. Then
\begin{equation}\nonumber
\begin{aligned}
&\eta\left(-\Delta_{A}(0)\right)=\eta\left(-\frac{p-1}{2}\cdot(p-1)\right)=\eta(-2)=1,&\\
&\eta\left(-\Delta_{A}\left(\frac{p-1}{2}\right)\right)=\eta\left(\frac{p-1}{2}\cdot\frac{p-1}{2}\right)=1,&\\
&\eta(-\Delta_{A}(p-1))=\eta\left(-(p-1)\frac{p-1}{2}\right)=\eta(-2)=1.&
\end{aligned}
\end{equation}
Thus, by Lemma \ref{fS}, there exists $V=(v_{1},v_{2},v_{3})\in(\mathbb{F}_{q}^{*})^{3}$  such that the code $\mathbf{EGRS}_{2}(A,V,q)$ is self-dual.
By Lemma \ref{Th2}, there exists a $q$-ary MDS  self-dual code of length $3p^{l}+1$.
\end{proof}

\begin{theorem}\label{Th}  Let $q=p^{m}$ be an odd prime power. Suppose $\eta(-1)=\eta(N)=1$ for any $N\in[2,t]$ with $2\leq t\leq p-1$. Then
\begin{itemize}
  \item[(1)] For $t$ odd, there exists a $q$-ary MDS self-dual code of length $(t+1)p^{l}$ with $0\leq l < m$.
  \item[(2)] For $t$ even, there exists a $q$-ary MDS self-dual code of length $(t+1)p^{l}+1$ with $0\leq l < m$.
\end{itemize}
\end{theorem}
\begin{proof}
We choose $A=\{0,1,2,\cdots,t\}$.
Since $\eta(-1)=\eta(N)=1$ for any $N\in[2,t]$. Then
$$\eta(\Delta_{A}(a))=\eta(-\Delta_{A}(a))=1 \ \mbox{for \ all} \ a\in A.$$
Thus, the conclusion is derived from Lemmas \ref{fS}, \ref{Th1} and \ref{Th2}.
\end{proof}

\begin{corollary}\label{CC1}  Let prime $p\equiv1 \ (\mbox{mod} \ 24)$ and $q=p^{m}$. Then there exists a $q$-ary MDS  self-dual code of length $5p^{l}+1$ with $0\leq l < m$.
\end{corollary}
\begin{proof}
Keep the notations as Theorem \ref{Th}. Let $t=4$. Since $p\equiv1 \ (\mbox{mod} \ 24)$, the result follows from
$$\eta(-1)=\eta(2)=\eta(3)=\eta(4)=1.$$
\end{proof}

\begin{remark} Let $t$ be a positive integer and $M=\mathrm{lcm}(1,2,\cdots,t,8)$. If $p\equiv1 \ (\mbox{mod} \ M)$, then it is clear that
$\eta(-1)=\eta(N)=1$ for any $N\in[2,t]$. The Dirichlet density of the set \[D=\left\{p \;\mathrm{is\ prime}\ | \ p\equiv1 \ (\mbox{mod} \ M)\right\}\] is equal to $\frac{1}{\phi(M)}$ where $\phi(\cdot)$ is the $Euler^{,}s \  totient$ function. Therefore, there are infinite numbers of prime $p$ satisfying Theorem \ref{Th}.
\end{remark}

Now we consider the union of cosets from multiplicative subgroup of $\mathbb{F}^{*}_{q}$. For brevity, 
\begin{itemize}
  \item let $\theta$ be a generator (primitive element) of $\mathbb{F}^{*}_{q}$.
  \item $q-1=e_{1}e_{2}$.
  \item $H_{1}=\langle\theta^{e_{1}}\rangle$ and $H_{2}=\langle\theta^{e_{2}}\rangle$.
  \item $\nu(a)=\min\{ x\in\mathbb{N} \ | \ a=\theta^{xe_{1}} \}$ for $a\in H_{1}$.
\end{itemize}

\begin{lemma}\label{Th3} Suppose $\mathbf{GRS}_{t}(A,V,q)$ is self-dual for some $A=\{a_{1},a_{2},\cdots,a_{2t}\}\subseteq H_{1}$ and $V\in(\mathbb{F}_{q}^{*})^{2t}$.
Then there exists a $q$-ary MDS  self-dual code of length $2te_{1}$, if $e_{1},e_{2}$ and $\nu(a_{i})$ satisfy one of the following conditions.
\begin{itemize}
  \item[(1)] $e_{1}$ is odd.
  \item[(2)] $e_{1}$ and $e_{2}$ are even, and $\nu(a_{i})(1\leq i\leq 2t)$ have the same parity.
\end{itemize}
\end{lemma}
\begin{proof}
 Label the elements of $B:=\bigcup\limits_{i=1}^{2t}(\theta^{\nu(a_{i})}H_{2})$ by $b_{1},b_{2},\cdots,b_{n}$. By Lemma \ref{YZF}, $$f_{\theta^{\nu(a_{i})}H_{2}}(x)=x^{e_{1}}-\theta^{\nu(a_{i})e_{1}} \quad \mbox{and} \quad \Delta_{\theta^{\nu(a_{i})}H_{2}}(x)=f^{'}_{\theta^{\nu(a_{i})}H_{2}}(x)=e_{1}x^{e_{1}-1}.$$
 If $b_{i}\in \theta^{\nu(a_{k})}H_{2}$ for some $k$, then there exists an integer $u\in[0,e_{1}-1]$ such that $b_{i}=\theta^{\nu(a_{k})+e_{2}u}$. Thus
\begin{equation}\nonumber
\begin{aligned}
\Delta_{B}(b_{i})&=\Delta_{\theta^{\nu(a_{k})}H_{2}}(b_{i})\prod_{j=1,j\neq k}^{2t}f_{\theta^{\nu(a_{j})}H_{2}}(b_{i})&\\
&=e_{1}\theta^{\nu(a_{k})(e_{1}-1)}\theta^{-e_{2}u}\prod_{j=1,j\neq k}^{2t}(\theta^{(\nu(a_{k})+e_{2}u)e_{1}}-\theta^{\nu(a_{j})e_{1}})&\\
&=e_{1}\theta^{\nu(a_{k})(e_{1}-1)}\theta^{-e_{2}u}\prod_{j=1,j\neq k}^{2t}(a_{k}-a_{j})&\\
&=e_{1}\theta^{\nu(a_{k})(e_{1}-1)}\theta^{-e_{2}u}\Delta_{A}(a_{k}).&
\end{aligned}
\end{equation}
(1) The fact $e_{1}$ is odd yields $e_{2}$ is even. Note that all of $\eta(\Delta_{A}(a_{i}))$ are equal, which implies all of $\eta(\Delta_{B}(b_{i}))(1\leq i \leq2te_{1})$ take the same value. By Lemma \ref{fS}, there exists a $q$-ary MDS  self-dual code of length $2te_{1}$.\\
(2) Since $e_{1},e_{2}$ are even, then
$$\eta(\Delta_{A}(a_{k})e_{1}\theta^{\nu(a_{k})(e_{1}-1)}\theta^{-e_{2}u})=\eta(\Delta_{A}(a_{k})e_{1}\theta^{\nu(a_{k})}).$$
Note that all of $\nu(a_{i})(1\leq i\leq 2t)$ have the same parity,
which implies all of $\eta(\Delta_{B}(b_{i}))(1\leq i \leq2te_{1})$ are equal. By Lemma \ref{fS}, there exists a $q$-ary MDS  self-dual code of length $2te_{1}$.
\end{proof}

\begin{remark} Now we give a short proof of (\cite{Yan}, Theorem 1(i)). Let $q\equiv1 \ (\mbox{mod} \ 4)$ be an odd prime power and let $n\mid q-1$ be an even positive integer. Then there exists a $q$-ary MDS  self-dual code of length $n<q-1$.
\end{remark}
\begin{proof}
Keep the same notations as Lemma \ref{Th3}. Let $q-1=2^{k}r$ with $\gcd(2,r)=1$. For $n=2^{k'}r^{'}$ with $\gcd(2,r^{'})=1$, we have $k^{'}\leq k$ and $r^{'}\mid r$.
\begin{itemize}
  \item If $k^{'}<k$, then we take $e_{1}=r^{'}$ and $A=\langle\theta^{2^{k-k^{'}}r}\rangle\subseteq\langle\theta^{e_{1}}\rangle$.
  \item If $k^{'}=k$, then $\frac{r}{r^{'}}\geq3$. We take $e_{1}=2^{k-1}r^{'}$ and $A=\{\theta^{e_{1}},\theta^{3e_{1}}\}\subseteq\langle\theta^{e_{1}}\rangle$.
\end{itemize}
\end{proof}

\begin{theorem}\label{GRS2} Let $q=p^{m}$ be an odd prime power and $q-1=e_{1}e_{2}$. Suppose $p\equiv1,3 \ (\mbox{mod} \ 8)$. For $e_{1}$ odd and $e_{2}\geq4$, there exists a $q$-ary MDS  self-dual code of length $4e_{1}$.
\end{theorem}
\begin{proof}
Let $\theta$ be a generator of $\mathbb{F}^{*}_{q}$. Since $p\equiv1,3 \ (\mbox{mod} \ 8)$,  $\eta(-2)=1$. We choose $A=\{1,-1,\theta^{e_{1}},\theta^{-e_{1}}\}\subseteq \langle\theta^{e_{1}}\rangle$. Then
\begin{equation}\nonumber
\begin{aligned}
&\eta\left(\Delta_{A}(1)\right)=\eta\left(2\cdot(1-\theta^{e_{1}})\cdot(1-\theta^{-e_{1}})\right)=\eta(-2\theta),&\\
&\eta(\Delta_{A}(-1))=\eta\left(-2\cdot(-1-\theta^{e_{1}})\cdot(-1-\theta^{-e_{1}})\right)=\eta(-2\theta),&\\
&\eta\left(\Delta_{A}\left(\theta^{e_{1}}\right)\right)=\eta\left((\theta^{e_{1}}-1)\cdot(\theta^{e_{1}}+1)
\cdot(\theta^{e_{1}}-\theta^{-e_{1}})\right)=\eta(\theta),&\\
&\eta\left(\Delta_{A}\left(\theta^{-e_{1}}\right)\right)=\eta\left((\theta^{-e_{1}}-1)\cdot(\theta^{-e_{1}}+1)
\cdot(\theta^{-e_{1}}-\theta^{e_{1}})\right)=\eta(\theta).&
\end{aligned}
\end{equation}
Thus, by Lemma \ref{fS}, there exists $V=(v_{1},v_{2},v_{3},v_{4})\in(\mathbb{F}_{q}^{*})^{4}$  such that the code $\mathbf{GRS}_{2}(A,V,q)$ is self-dual.
By Lemma \ref{Th3}, there exists a $q$-ary MDS  self-dual code of length $4e_{1}$.
\end{proof}

\begin{lemma}\label{111} Let $A=\{a_{1},a_{2},\cdots,a_{t}\}\subseteq H_{1}$.
\begin{itemize}
  \item[(1)] If $t$ and $e_{1}$ are both odd and $\eta(\prod\limits_{i=1}^{t}a_{i})=\eta(-e_{1}\Delta_{A}(a_{i})a_{i})$ for $1\leq i \leq t$, then there exists a $q$-ary MDS  self-dual code of length $te_{1}+1$.
  \item[(2)] If one of $t,e_{1}$ is even, $\eta(e_{1})=\eta(-\Delta_{A}(a_{i})a_{i})(1\leq i \leq t)$ and $\eta((-1)^{t+1}\prod\limits_{i=1}^{t}a_{i})=1$, then there exists a $q$-ary MDS  self-dual code of length $te_{1}+2$.
\end{itemize}
\end{lemma}
\begin{proof}
Label the elements of $B:=\bigcup\limits_{a_{i}\in A}(\theta^{\nu(a_{i})}H_{2})\cup \{0\}$ by $b_{1},b_{2},\cdots,b_{n}$. By Lemma \ref{YZF}, $$f_{\theta^{\nu(a_{i})}H_{2}}(x)=x^{e_{1}}-\theta^{\nu(a_{i})e_{1}} \quad \mbox{and} \quad \Delta_{\theta^{\nu(a_{i})}H_{2}}(x)=f^{'}_{\theta^{\nu(a_{i})}H_{2}}(x)=e_{1}x^{e_{1}-1}.$$
If $0\neq b_{i}\in \theta^{\nu(a_{k})}H_{2}$ for some $k$, then there exists an integer $u\in[0,e_{1}-1]$ such that $b_{i}=\theta^{\nu(a_{k})+e_{2}u}$. Thus
\begin{equation}\nonumber
\begin{aligned}
\Delta_{B}(b_{i})&=\Delta_{\theta^{\nu(a_{k})}H_{2}}(b_{i})b_{i}\prod_{j=1,j\neq k}^{t}f_{\theta^{\nu(a_{j})}H_{2}}(b_{i})&\\
&=e_{1}\theta^{\nu(a_{k})e_{1}}\prod_{j=1,j\neq k}^{t}(\theta^{(\nu(a_{k})+e_{2}u)e_{1}}-\theta^{\nu(a_{j})e_{1}})&\\
&=e_{1}a_{k}\prod_{j=1,j\neq k}^{t}(a_{k}-a_{j})&\\
&=\Delta_{A}(a_{k})a_{k}e_{1}&
\end{aligned}
\end{equation}
and $\eta(\Delta_{B}(0))=\eta((-1)^{t}\prod\limits_{i=1}^{t}a_{i})$.\\
(1) For $t$ odd and  $\eta(\prod\limits_{i=1}^{t}a_{i})=\eta(-e_{1}\Delta_{A}(a_{i})a_{i})$ for $1\leq i \leq t$, we have that
all of $\eta(-\Delta_{B}(b_{i}))(1\leq i \leq te_{1}+1)$ are equal.
By Lemma \ref{fS}, there exists a $q$-ary MDS  self-dual code of length $te_{1}+1$.\\
(2) If one of $t,e_{1}$ is even, $\eta(e_{1})=\eta(-\Delta_{A}(a_{i})a_{i})(1\leq i \leq t)$ and $\eta((-1)^{t+1}\prod\limits_{i=1}^{t}a_{i})=1$, then $\eta(-\Delta_{B}(b_{i}))=1(1\leq i \leq te_{1}+1)$.
By Lemma \ref{fS}, there exists a $q$-ary MDS  self-dual code of length $te_{1}+2$.
\end{proof}

\begin{remark} Let $A=\{1\}$.
\begin{itemize}
  \item[(1)] If $e_{1}$ is odd and $\eta(-e_{1})=1$, then there exists a $q$-ary MDS  self-dual code of length $e_{1}+1$ (see \cite{Yan}, Theorem 1(ii)).
  \item[(2)] If $e_{1}$ is even and $\eta(-e_{1})=1$, then there exists a $q$-ary MDS  self-dual code of length $e_{1}+2$ (see \cite{Yan}, Theorem 1(iii)).
\end{itemize}
\end{remark}

\begin{theorem}\label{Th4} Let $q=p^{m}$ be an odd prime power and $s\mid m$. Then there exists a $q$-ary MDS self-dual code of length $n$ as follows.
\begin{itemize}
  \item[(1)] $n=(t+1)p^{sl}$ with $0\leq l < \frac{m}{s}$, if $t$ is odd, $t\mid (p^{s} -1)$ and $\eta(-t)= 1$.
  \item[(2)] $n=(t+1)p^{sl}+1$ with $0\leq l < \frac{m}{s}$, if $t$ is even, $t\mid (p^{s}-1)$ and $\eta(t)=\eta(-1)=1$.
\end{itemize}
\end{theorem}
\begin{proof} Let $\theta$ be a generator of $\mathbb{F}^{*}_{q}$. Then $\mathbb{F}^{*}_{p^{s}}=\left\langle\theta^{\frac{q-1}{p^{s}-1}}\right\rangle$. We take $e_{1}=t$ and $A=\{1\}$.
\begin{itemize}
  \item[(1)] Since $t$ is odd and $\eta(-t)= 1$, by Lemma \ref{111}, $\mathbf{GRS}_{\frac{t+1}{2}}(B,V,q)$ is self-dual for some $V\in(\mathbb{F}_{q}^{*})^{t+1}$ where $B=\left\langle\theta^{\frac{q-1}{t}}\right\rangle\cup\{0\}$. Since $t\mid (p^{s}-1)$, we have $B\subseteq\mathbb{F}_{p^{s}}$. By Lemma \ref{Th1}, there exists a $q$-ary MDS  self-dual code of length $(t+1)p^{sl}$ with $0\leq l < \frac{m}{s}$.
  \item[(2)] Since $t$ is even and $\eta(t)=\eta(-1)= 1$, by Lemma \ref{111}, $\mathbf{EGRS}_{\frac{t+2}{2}}(B,V,q)$ is self-dual for some $V\in(\mathbb{F}_{q}^{*})^{t+2}$ where $B=\left\langle\theta^{\frac{q-1}{t}}\right\rangle\cup\{0\}$. Note that $t\mid (p^{s} -1)$. Thus, $B\subseteq\mathbb{F}_{p^{s}}$. Since $\eta(-1)=1$ yields $q\equiv1\,(\mbox{mod}\,4)$, by Lemma \ref{Th2}, there exists a $q$-ary MDS  self-dual code of length $(t+1)p^{sl}+1$ with $0\leq l < \frac{m}{s}$,
\end{itemize}
\end{proof}

\section{Conclusion}

The criterions of MDS  self-dual codes is given in \cite{JX2,Yan}.  A. Zhang and K. Feng \cite{ZF} considered constructions of MDS  self-dual codes from small filed. In this paper, the constructions of MDS  self-dual codes from short length is considered. The proof of the results were concise by using the notations of \cite{ZF}. Furthermore, some known results can be considered as special cases in our results.  The exact constructions of MDS  self-dual codes with short length (especially $n=3,4,5,6$) are given so that we obtained new  MDS  self-dual codes. Note that there are only a few known results about $q\equiv3 \ (\mbox{mod} \ 4)$. For $q\equiv3 \ (\mbox{mod} \ 8)$, we given MDS  self-dual codes of length $3p^{l}+1$ and $4e_{1}$ in Theorem \ref{EGRS} and \ref{GRS2}. Finally, combining  with Lemmas \ref{Th1}-\ref{111} and known results, we can get new $q$-ary MDS self-dual code.

%======================================== 参考文献部分 ======================================


\begin{thebibliography}{1}
\vskip2mm
\bibitem{FF3} W. Fang and F. Fu, ``New constructions of MDS Euclidean self-dual codes from GRS codes and extended GRS codes," \emph{IEEE Trans. Inf. Theory}, vol. 65, no. 9, pp. 5574--5579, 2019.
\bibitem{FLL} X. Fang, M. Liu, and J. Luo, ``New MDS Euclidean Self-orthogonal Codes," arXiv: 1906.00380 [cs.IT], Sep. 2019.
\bibitem{FLLL} X. Fang, K. Lebed, H. Liu, and J. Luo, ``New MDS Euclidean self-dual codes over finite fields of odd characteristic," arXiv:1811.02802v9 [cs.IT], Sep. 2019.

\bibitem{GK} S. Georgion and C. Koukouvinos, ``MDS Euclidean self-dual codes over large prime fields," \emph{Finite Fields and Their Appl.}, vol. 8, no. 4, pp. 455-470, 2002.

\bibitem{GG} M. Grassl and T. A. Gulliver, ``On self-dual MDS codes," in \emph{Proc. of ISIT}, pp. 1954-1957, 2008.

\bibitem{GKL} T. A. Gulliver, J. L. Kim, and Y. Lee, ``New MDS or near-MDS Euclidean self-dual codes," \emph{IEEE Trans. Inf. Theory}, vol. 54, no. 9, pp. 4354-4360, 2008.

\bibitem{HK1} M. Harada and H. Kharaghani, ``Orthogonal designs, self-dual codes and the Leech lattice," \emph{J. Combin. Designs}, vol. 13, no. 3, pp. 184-194, 2005.

\bibitem{HK2} M. Harada and H. Kharaghani, ``Orthogonal designs and MDS Euclidean self-dual codes," \emph{Australas. J. Combin.}, vol. 35, pp. 57-67, 2006.
\bibitem{JX2} L. Jin and C. Xing, ``New MDS self-dual codes from generalized Reed-Solomon codes," \emph{IEEE Trans. Inf. Theory}, vol. 63, no. 3, pp. 1434-1438, 2017.

\bibitem{KZT} X. Kai, S. Zhu S, and Y. Tang, ``Some constacyclic self-dual codes over the integers modulo $2^{m}$," \emph{Finite Fields and Their Appl.} vol. 18, no. 2, pp. 258-270, 2012.

\bibitem{KL1} J. L. Kim and Y. Lee, ``MDS Euclidean self-dual codes," in \emph{Proc. of ISIT}, pp. 1872-1877, 2004.

\bibitem{KL2} J. L. Kim and Y. Lee, ``Euclidean and Hermitian self-dual MDS codes over large finite fields," \emph{J. Combin. Theory, Series A}, vol. 105, no. 1, pp. 79-95, 2004.

\bibitem{LLL} K. Lebed, H. Liu, and J. Luo, ``Construction of MDS Euclidean self-dual codes over finite field," \emph{Finite Fields and Their Appl.}, vol. 59, pp. 199-207, 2019.
\bibitem{MS} F. J. MacWilliams and N. J. A. Sloane, \emph{The Theory of Error-correcting Codes}. The Netherlands: North Holland, Amsterdam, 1977.
\bibitem{TW} H. Tong and X. Wang, ``New MDS Euclidean and Herimitian self-dual codes over finite fields," \emph{Advances in Pure Mathematics.}, vol. 7, no. 5, 2016.

\bibitem{Yan} H. Yan, ``A note on the construction of MDS Euclidean self-dual codes," \emph{Cryptogr. Commun.}, vol. 11, no. 2, pp. 259-268, 2019.

\bibitem{YC} Y. Yang and W. Cai, ``On self-dual constacyclic codes over finite fields," \emph{Des. Codes Cryptogr.}, vol. 74, no. 2, pp. 355-364, 2015.

\bibitem{ZF} A. Zhang and K. Feng, ``An unified approach on constructing of MDS self-dual codes via Reed-Solomon codes," arXiv: 1905.06513v1 [cs.IT], May 2019.

\end{thebibliography}
\end{document}